\documentclass[12pt]{article}

\usepackage{amsmath}
\usepackage{amsfonts,amssymb,amsthm,overpic}
\makeatletter
\newcommand{\diam}{\mathop{\operator@font diam}}
\usepackage{setspace}
\usepackage{multicol}
\usepackage{multirow}
\usepackage[compact]{titlesec}

\usepackage{epsfig}

\newtheorem{definition}{Definition}[section]

\newtheorem{theorem}{Theorem}[section]

\newtheorem{lemma}{Lemma}[section]

\newcommand{\cT}{\mathcal{T}}

\begin{document}

\title{\Huge{\textsc{Spacetime Singularities vs. Topologies in the Zeeman-G\"obel Class }}}

\author{Kyriakos Papadopoulos$^1$, Basil K. Papadopoulos$^2$ \\
\small{1. Department of Mathematics, Kuwait University, PO Box 5969, Safat 13060, Kuwait} \\
\small{2. Department of Civil Engineering, Democritus University of Thrace, Greece}}

\date{}

\maketitle

\begin{abstract}
In this article we first observe that the Path topology of Hawking, King and MacCarthy is an analogue, in curved spacetimes, of a topology that was suggested by Zeeman as an alternative topology to his so-called Fine topology in Minkowski spacetime. We then review a result of a recent paper on spaces of paths and the Path topology, and see that there are at least five more topologies in the class $\mathfrak{Z}-\mathfrak{G}$ of Zeeman-G\"obel topologies which admit a countable basis, incorporate the causal and conformal structures, but the Limit Curve Theorem fails to hold. The "problem" that L.C.T. does not hold can be resolved by "adding back" the light-cones in the basic-open sets of these topologies, and create new basic open sets for new topologies. But, the main question is: do we really need the L.C.T. to hold, and why? Why is the manifold topology, under which the group of homeomorphisms of a spacetime is vast and of no physical significance (Zeeman), more preferable from an appropriate topology in the class $\mathfrak{Z}-\mathfrak{G}$ under which a homeomorphism is an isometry (G\"obel)? Since topological conditions that come as a result of a causality requirement are key in the existence of singularities in general relativity, the global topological conditions that one will supply the spacetime manifold might play an important role  in describing the transition from the quantum non-local
theory to  a classical local theory.
\end{abstract}

\section{Motivation.}

In paper \cite{Ordr-Ambient-Boundary} the authors talked about the causal structure of the ``ambient boundary'',
which is the conformal infinity of a five dimensional ``ambient space''. The whole deal of this construction
was to put relativity theory into a wider frame, in order to understand better the nature of the spacetime
singularities. Actually, there is an unexplored ocean in the literature of singularities: do they belong
to the spacetime or not? Or, better, are there singularities which do not belong to the spacetime and
others which do belong to? Of course, as soon as there is spacetime, there are events and, for every event,
there is a null cone and a Lorentzian metric. Having this in mind, Zeeman, G\"obel and Hawking-King-McCarthy
built natural topologies, which incorporate the causal, differential and conformal structures, against
the manifold topology which is alien to the most fundamental properties of spacetime, for example the
structure of the null cone. From a mathematical point of view, a spacetime is not a complete model, if
it is not equipped with an appropriate ``natural'' topology. By natural we mean a topology compatible
with its most fundamental structures. Having in mind that the manifold topology is not natural, one
can easily ask the question: why do we use such a topology to prove the validity of theorems
like the Limit Curve Theorem, while other, much more natural topologies, are ignored.

\section{Preliminaries.}

Throughout the text, unless otherwise stated, by $M$ we will denote any spacetime manifold and not, restrictively,
the Minkowski space. In \cite{Zeeman1} the null-cones where viewed in a topological context, so when
one constructs a topology in a spacetime incorporating its causal structure, one will consider the time-cone of an event (interior of the null-cone), the light-cone (the boundary),
 the causal-cone (interior union boundary) and the complement of the causal-cone (exterior); Zeeman calls this exterior ``space-cone'' and he proposes (in \cite{Zeeman1})
an alternative topology based on such space-cones, in its construction. Since the results of Zeeman have been
proven to extend to any curved spacetime by G\"obel (see \cite{gobel}), the construction
of topologies which incorporate the causal structure is independent on whether
the null-cones are affected by the curvature of the spacetime. As soon as there is spacetime there
are events and as soon there is an event there is a null-cone; it is its interior, boundary
and exterior which will play a role in a topological construction, and not the linear structure
of a spacitme like the Minkowski space or the curvature of a spacetime manifold.

\subsection{Causality.}

In spacetime geometry, one can introduce three causal relations, namely, the
chronological order $\ll$, the causal order $\prec$ and the relation horismos $\rightarrow$, and these can be meaningfully extended to  any \emph{event-set}, a set
$(X,\ll,\prec,\rightarrow)$ equipped with all three of these orders having no metric \cite{Penrose-Kronheimer,Penrose-difftopology}. In this context we say that the event $x$ chronologically precedes an event $y$, written
$x\ll y$ if $y$ lies inside the future null cone of $x$, $x$ {\em causally precedes} $y$,
$x \prec y$, if $y$ lies inside or on the future null cone of $x$ and  $x$ is at {\em horismos}
with $y$, written $x \rightarrow y$, if $y$ lies on the future null cone of $x$. The chronological order is irreflexive, while the causal order and horismos are reflexive. Then, the notations $I^+(x) = \{y \in M : x \ll y\},  J^+(x) = \{y \in M : x \prec y\}$ will be used for the chronological and the causal futures  of $x$ respectively (and with a minus instead of a plus sign  for the pasts), while the future null cone of $x$ will be denoted by $\mathcal{N}^+(x)\equiv\partial J^{+}(x)= \{y \in M : x \rightarrow y\}$, and dually for the null past of $x$, cf. \cite{Penrose-difftopology}.

The above definitions of  futures and pasts of a set can be trivially extended to the situation of any partially ordered set $(X,<)$.
In a purely topological context this is usually done by passing to the so-called upper (i.e. future) and lower (i.e. past)
sets which in turn lead to the  future  and past  topologies (see \cite{Compendium}, for the special case of a lattice;
here our topologies are constructed in a similar manner, but are weaker, since they do not depend on lattice
orders, but on weaker relations).
A subset $A \subset X$ is a {\em past set} if $A = I^-(A)$ and dually for the future. Then, the {\em future topology} $\cT^+$ is generated
by the subbase $\mathcal{S}^+ = \{X \setminus I^-(x) : x \in X\}$
and the {\em past topology} $\cT^-$ by $\mathcal{S}^- = \{X \setminus I^+(x)  : x \in X\}$.
The {\em interval topology} $\cT_{in}$ on $X$ then consists of basic sets which are finite intersections
of subbasic sets of the past and the future topologies. This is in fact the topology that fully characterises a given order of the poset $X$. Here we clarify that the names ``future topology'' and ``past topology'' are the best possible inspirations for names that
came in the mind of the authors, but are new to the literature. The motivation was
that they are generated by
complements of past and future sets, respectively (i.e. closures of future and
past sets, respectively). Also, the authors wish to highlight the distinction between
the interval topology $\cT_{in}$ which appears to be of an important significance in
lattice theory, from the ``interval topology'' of A.P. Alexandrov (see \cite{Penrose-difftopology}, page 29).
$\cT_{in}$ is of a more general nature, and it can be defined via any relation, while the Alexandrov
topology is restricted to the chronological order. These two topologies are different in nature,
as well as in definition, so we propose the use of ``interval topology'' for $\cT_{in}$ exclusively,
and not for the Alexandrov topology. It is worth mentioning that the Alexandrov topology being
Hausforff is equivalent to the Alexandrov topology being equal to the manifold topology which is equivalent
to the spacetime being strongly causal; the topologies that we mention is this paper are
not equal to the manifold topology.

The so-called \emph{orderability  problem} is concerned with the conditions under which the topology $\cT_<$ induced
by the order $<$  is equal to some given topology $T$ on $X$ (\cite{Good-Papadopoulos}, \cite{Orderability-Theorem},
\cite{OnProperties} and \cite{Nestsandtheirrole}). In \cite{On-Two-Zeeman-Topologies} and \cite{Order-Light-Cone}
the authors found specific solutions of the orderability problem for six distinct  topologies in the class
of Zeeman-G\"obel. In these article, we consider interval topologies which, together with the manifold topology,
produce intersection topologies that belong to the clas of Zeeman-G\"obel and under which the
Limit Curve Theorm fails. Here we remind the definition of intersection topology.

\begin{definition}\label{intersection topology}
If $T_1$ and $T_2$ are two distinct topologies on a set $X$, then the {\em intersection topology} $T^{int}$
with respect to $T_1$ and $T_2$, is the topology on $X$ such that the set $\{U_1 \cap U_2 : U_1 \in T_1, U_2 \in T_2\}$
forms a base for $(X,T)$.
\end{definition}

and the following useful lemma (see \cite{Order-Light-Cone}):

\begin{lemma}\label{before main theorem}

Let $T_1$ and $T_2$ be two topologies on a set $X$, with bases $\mathcal{B}_1$ and $\mathcal{B}_2$ respectively and
let $T^{int}$ be their intersection topology, provided that it exists. Then,
the following two hold.

\begin{enumerate}

\item The collection $\mathcal{B}^{int} = \{B_1 \cap B_2 : B_1 \in \mathcal{B}_1, B_2 \in \mathcal{B}_2\}$ forms
a base for $T^{int}$.

\item If $\mathcal{B}^{int}$ is a base for a topology, then this topology is $T^{int}$.

\end{enumerate}
\end{lemma}

It is rather surprising, as we shall see, that the interval topology, generated from
either the chronological order, the relation horismos or a spacelike non-causal order that
we will define shortly, is one of the two constituents for several Zeeman-G\"obel topologies
which are actually intersection topologies,
the other constituent being the manifold topology.

\subsection{The Class of Zeeman-G\"obel Topologies.}

The class $\mathfrak{Z}$, of Zeeman topologies, is the class of topologies
on the Minkowski space $M$ strictly finer than the Euclidean topology and strictly coarser
than the discrete topology,
which have the property that they induce the $1$-dimensional Euclidean topology on every time
axis and the $3$-dimensional Euclidean topology on every space
axis.

Zeeman (see \cite{Zeeman2} and \cite{Zeeman1}) showed that
the causal structure of the null cones on the Minkowski
space determines its linear structure. After initiating
the question on whether a topology on Minkowski space, which
depends on the light cones, implies its linear structure as well,
he constructed the Fine Topology, $F$, which is defined as the finest topology for $M$ in the class
$\mathfrak{Z}$.

$F$ satisfies, among other properties, the
following two theorems:

\begin{theorem}\label{Theorem2}
Let $f : I \to M$ be a continuous map of the unit interval $I$
 into $M$. If $f$ is strictly $\ll$-preserving, then the image $f(I)$ is a
 piecewise linear path, consisting of a finite number of intervals along time axes.
 \end{theorem}

Let $G$ be the group of automorphisms of $M$, given by the Lorentz group,
translations and dilatations.

\begin{theorem}\label{main}
The group of homeomorphisms of the Minkowski space under $F$ is $G$.
\end{theorem}

G\"obel (see \cite{gobel}) showed that the results of Zeeman are valid
without any restrictions on the spacetime, showing in particular that
the group of homeomorphisms of a spacetime $S$, with respect to the
general relativistic analogue of $F$, is the group of all homothetic
transformations of $S$.

G\"obel defined Zeeman topologies in curved spacetimes as follows.
Let $T_M$ be the manifold topology, let $M$ be
a spacetime manifold and let $S$ be a set of subsets of $M$. A set $A \subset M$
is open in $Z(S,T_M)$, a topology in class $\mathfrak{Z}-\mathfrak{G}$ of Zeeman-G\"obel, if $A \cap B$ is open
in $(B,T_M|B)$ (the subspace topology of the manifold topology $(M,T_M)$
with respect to $(B,T_M)$), for all $B \in S$.



\section{The Path Topology of Hawking-King-McCarthy is in $\mathfrak{Z}-\mathfrak{G}$.}

Zeeman in \cite{Zeeman1}, last section, introduced three alternatives to his Fine topology,
one of which was described as follows:
\begin{definition}\label{Zeeman_Path_Topology}
$Z^T$ is the finest topology that induces the $1$-dimensional Euclidean topology on
every time axis; an open neighbourhood of $x \in M$, in this topology, is given by $B_\epsilon(x) \cap C^T(x)$.
\end{definition}

Even the fact that there is no clear reference in \cite{Hawking-Topology} (or in later
papers, like \cite{Malament}), that the
path topology $\mathcal{P}$ is the analogue of $Z^T$ for curved spacetimes, this is
trivially deduced within the frame of point-set topology. Here we adopt the much clearer notation of Low in \cite{Low_path}:
for each $x \in M$ and each open neighbourhood $U$ of $x$, let $I(p,U)$ denote
the set of points connected to $p$ by a timelike path lying in $U$ and by
$K(p,U)$ the set $I(p,U) \cup \{x\}$. By choosing an arbitrary Riemannian metric
$h$ on $M$, let $B_\epsilon(x)$ denote an open ball centered at $x$ with radius
$\epsilon >0$, with respect to $h$.

\begin{definition}\label{Path_Topology}
The {\em path topology}, $\cal{P}$, is defined to be the finest topology so that the induced
topology on every timelike curve coincides with the topology induced from the manifold
topology.
\end{definition}

For a proof of the following theorem see \cite{Hawking-Topology}:

\begin{theorem}\label{hawking_basis_for_path}
Sets of the form $K(p,U) \cap B_\epsilon(x)$ form a basis for the topology $\cal{P}$.
\end{theorem}

In Section 1.3, from \cite{On-Two-Zeeman-Topologies}, we introduced a way to partition the space cone
$S(x)$, for an event $x$ in $M$, into two symmetrical subcones
$S^+(x)$ and $S^-(x)$, and constructed a space-like (non-causal) order $<$,
such that $x <y$ if $y \in S^+(x)$ and $x>y$ if $y \in S^-(x)$.
By $x\leq y$ we meant that either $y \in S^+(x)$ or $x \rightarrow y$
and, respectively, $y \leq x$ if either $y \in S^-(x)$ or $x \rightarrow y$,
where $\rightarrow$ denoted the irreflexive version of horismos.
Using this material, we introduced the following theorem:

\begin{theorem}\label{1}
The order $\leq$ induces a topology $T_{in}^{\le}$ in $\mathfrak{Z}$, which is
an interval topology generated by the spacelike order $\le$ and,
furthermore, $Z^T$ is the intersection topology of $T_{in}^\le$ and
the topology $T_{\mathbb{R}^4}$ on the Minkowski space $M$. In addition, $T_{in}^{\le}$ (and, consequently, $Z^T$)
is defined invariantly of the choice of the partition of the null cone
into two, for defining $\le$.
\end{theorem}

We summarise our observations in the following theorem:

\begin{theorem}\label{path-is-Zeeman}
The analogue of $Z^T$, in curved spacetimes, is the path topology $\mathcal{P}$.
In particular, the order $\leq$ induces a topology $T_{in}^{\le}$ in $\mathfrak{Z} - \mathcal{G}$, which is
an interval topology generated by the spacelike order $\le$ and,
furthermore, $\mathcal{P}$ is the intersection topology of $T_{in}^\le$ and
the topology $T_M$ on the spacetime manifold $M$. In addition, $T_{in}^{\le}$ (and, consequently, $\mathcal{P}$)
is defined invariantly of the choice of the partition of the null cone
into two, for defining $\le$.
\end{theorem}

From Theorem \ref{path-is-Zeeman}, we deduce that the path topology $\mathcal{P}$ is locally an order topology,
whose order is non causal and spacelike. It is interesting that this spacelike order
generates open sets that are timelike; there is, somehow, a link
(a duality perhaps) between spacelike and timelike that could be
further explored.

\section{The L.C.T. fails under a $\mathcal{P}$ environment.}

R.J. Low showed (see \cite{Low_path}) that if in a curved spacetime we substitute
the manifold topology with $\mathcal{P}$, then the Limit Curve Theorem (L.C.T.) fails to hold.
We read, in the conclusion of this article, that this suggests that although the path topology
is of great interest from the point of view of encapsulating the differential and causal structure
of spacetime, it is nevertheless inappropriate for at least some important aspects of the study
of the causal structure, where the manifold topology remains both technically easier to work with
and fruitful. A natural question though
could be the following: why is it not fruitful that the L.C.T. fails, while it holds for a topology (the
manifold one) which,
as one can read in a list of arguments in \cite{Zeeman1} and \cite{gobel}, misses
important elements of the spacetime?

We will give next a list of topologies,
all in the class $\mathfrak{Z}-\mathfrak{G}$, where the L.C.T. fails. These topologies have a common characteristic:
 they all admit a countable basis and their basic-open sets depend on the causal structure of
the spacetime.

\section{Six Spacetime Topologies where the L.C.T. fails.}

Before we proceed, we find it important to mention the version of L.C.T. given in \cite{Low_path}
(for a further reading see \cite{Geroch}): \\ If $\gamma_n$ is a sequence of causal paths and
$x_n \in \gamma_n$ with $x$ a limit point of $\{x_n\}$, then there is an endless causal path
$\gamma$ through $x$, which is a limit curve of $\{\gamma_n\}$, all within the frame of the
manifold topology.

We should also mention that by the relation $\ll^=$ we denote
the irreflexive causal order $\prec$, i.e.
$x \ll^= y$ if either $y$ lies on the future timecone of $x$
or future lightcone of $x$, but $y$ cannot be equal to $x$. By
$\rightarrow^{irr}$ we define irreflexive horismos, that is,
$x \rightarrow^{irr} y$, iff $y$ lies on the future lightcone
of $x$ but $y$ cannot be equal to $x$.

\begin{theorem}\label{2}
There are six topologies, in a spacetime manifold, which admit a countable basis, they
incorporate the causal structure and the L.C.T. fails with each one of them respectively.
These are the intersection topologies $Z$, $Z^T$, $Z^S$ and the interval topologies
$T_{in}^{{\rightarrow^{irr}}}$, $T_{in}^{\le}$ and $T_{in}^{\ll^=}$, which are all in the class
$\mathfrak{Z} - \mathfrak{G}$.
\end{theorem}

\begin{proof}
The construction of the topologies $Z$ and $T_{in}^{\rightarrow}$ is described analytically in \cite{Order-Light-Cone}
and of the rest four in \cite{On-Two-Zeeman-Topologies}. By looking at the arguments of
R.J. Low in \cite{Low_path} (section V., straight after Proposition 6), for each
of the six topologies individually, we see that in
each case, if $U$ is an open set in the assigned topology, not containing
the origin, and $p$ is an event in $\gamma \cap U$, then $\gamma$
is not a limit curve of $\{\gamma_n\}$, under the assigned topology. The uniqueness
of $\gamma$ implies the failure of the L.C.T.
\end{proof}

In particular, $Z^T = \mathcal{P}$ is the intersection topology between $T_{in}^{\le}$ and
the manifold topology $T_M$; $Z^T$ is the intersection topology between $T_{in}^{\ll^=}$ and $T_M$, while $Z$ is
the intersection topology between
$T_{in}^{{\rightarrow^{irr}}}$ and $T_M$. All these topologies can be constructed
in a spacetime independently from whether we are talking about Minkowski space or a curved spacetime
(since they depend on topological characteristics of the nullcone),
so we do not consider it important to use deferent notations for the general relativistic
against the special relativistic analogues.

The basic-open sets of $T_{in}^{\le}$ are timecones and of $T_{in}^{\ll^=}$
spacecones, so there is a duality between the topologies $\mathcal{P}$ and $Z^S$.
Also, $T_{in}^{\ll^=}$ is generated by an irreflexive causal relation,
so we again see another possible duality; a causal order generates a topology whose
basic-open sets are spacecones, a similar pattern that we mentioned about $\mathcal{P}$.

\section{Four Candidate Spacetime Topologies passing the L.C.T.}

It can be easily seen that the L.C.T., in the topologies of the previous
section, failed because, in all cases, the light cone was extracted when forming
the basic open sets. What about if the light-cone was ``replaced back'', and considered
the causal-cone (time-cone union light-cone) for forming basic-open sets? In this case,
$Z'$ (the topology which would have as basic-open sets those made of by
 basic-open sets of $Z$ adding the lightcone) would become the usual manifold topology with the Riemannian balls and $[T_{in}^{\rightarrow}]'$ would
be the indiscrete topology. We are not interested for these two cases,
especially for the indiscrete one. Now, $[Z^T]'$ would have as basic
open sets the bounded causal cones (time-cone union light cone, intersected
with a ball from the manifold topology) and $[T_{in}^{\le}]'$ would have as basic
open sets the causal cones. Similarly for $[Z^S]'$: its basic open sets would
be the bounded closure of space-cones ),
i.e. space-cone union light cone, intersected with a ball from the manifold topology.
In $[T_{in}^{\ll^=}]'$, the basic open sets would be unbounded closures of space-cones.

\begin{theorem}\label{3}
There are four topologies, in a spacetime manifold, which admit a countable basis, they
incorporate the causal structures and the L.C.T. holds for each one of them respectively.
These are the topologies $(Z^T)'$, $(Z^S)'$,
 $(T_{in}^{\le})'$ and $(T_{in}^{\ll^=})'$.
\end{theorem}
\begin{proof}
All basic-open sets of the dashed topologies are containg the light cone, so a similar
argument to the proof of Theorem \ref{2} leads us to the conclusion
that the L.C.T. holds for each one of the four topologies, individually.
\end{proof}

A few more comments about the construction of the dashed topologies.  Both $(T_{in}^{\le})'$ and $(T_{in}^{\ll^=})'$
are interval topologies and $(Z^T)'$, $(Z^S)'$ intersection topologies.

To see, for example, if $(Z^T)'$ is the intersection topology of the interval topology $(T_{in}^{\le})'$ and the
manifold topology $T_M$ and, if so,
to find the order from which this interval topology is
induced, we first consider the complements of the sets:
\[ S^+(x) = \{y \in M : y > x\},~~~~~(1)\]
and dually for $M - S^-(x)$; these will be considered subbasic-open sets. Then, we observe that the intersection of the subbasic
sets $M - S^+(x)$ and $M - S^-(x)$ gives a $(T_{in}^{\le})'$ open set (a causal cone, that is,
the time cone at $x$ union the light cone)
and, consequently, a $Z^T$ open set if we intersect it with a ball $N^M_\epsilon(x)$ of $T_M$, because for each
$x$, $M - S^+(x) \cap M - S^-(x) \cap N^M_\epsilon(x)$ gives a neighbourhood $N^M_\epsilon(x)$ of some
radius $\epsilon$, under a Riemannian metric, with the space
cone removed, but $x$ is kept. The considerations for the remaining topologies are similar.

\section{A Few Questions.}

{\bf{Conjecture.}} We conjecture that Theorem \ref{2} exhausts all
possible spacetime topologies which belong to the class $\mathfrak{Z}-\mathfrak{G}$, admit
a countable basis, incorporate the causal and conformal structure
L.C.T. fails at each one, respectively. This is due to the fact that
these topologies have basic-open sets which depend on the structure
of the null cone and, in each case, the light-cone is subtracted
when forming the basic-open sets.

{\bf{Question.}} How many alternative spacetime topologies are there, that
admit a countable basis, they incorporate the causal structure and where
the L.C.T. holds at each one, respectively? There is a number of recent
articles, which define such topologies which do not have any obvious
link to the topologies in class $\mathfrak{Z}-\mathfrak{G}$, such as in \cite{Seth}
and \cite{Sumati}.

{\bf{Question.}} Changing topologies in the class $\mathfrak{Z}-\mathfrak{G}$ looks
like ``dressing'' the spacetime with different clothes, while the spacetime
itself remains the same; or, maybe not? Topologies in $\mathfrak{Z}-\mathfrak{G}$ are compatible
with the metric (see \cite{gobel}); in some of them,
the topologies in Theorem \cite{Zeeman2} lead to the conclusion that
it may not be possible to formulate the same sufficient conditions for
geodesic incompleteness as it is usually done using the manifold topology,
while for some others (including the Fine topology, in $\mathfrak{Z}-\mathfrak{G}$)
the L.C.T. holds. The question is: what is the physical interpretation
of this phenomenon?

\section{Implications in Quantum Gravity.}

The implications of the above considerations for gravitational theories might be considerable. We only comment on an issue that has been raised in the causal set program for quantum gravity. There is a growing sense about an underlying non-locality of gravitational
interactions at the quantum level.  The question that has to be addressed is how one makes the transition from the quantum non-local
theory to  a classical local theory. This was eloquently formulated in \cite{Rafael}. Since locality is fundamentally a  topological issue, it might be
possible to re-frame this question in topological terms as a ``phase transition" between two different topologies one operating at the
quantum and the other at the classical levels. How exactly this transition occurs, or if it occurs at all, can be the subject of future work.

\section*{Acknowledgement} The author K.P. would like to thank Robert Low for his comments
on the interval topology, Fabio Scardigli for the discussions on the topological structure
of a spacetime and Nikolaos Kalogeropoulos, for the inspiring
 discussions on quantum gravity, based on his article \cite{Kalogeropoulos}.


\begin{thebibliography}{99}
\bibitem{Penrose-Kronheimer} E.H. Kronheimer and R. Penrose, {\sl On the structure of causal spaces},
Proc. Camb. Phil. Soc. (1967), 63, 481.

\bibitem{Penrose-difftopology} R. Penrose, {\sl Techniques of Differential Topology in Relativity},
CBMS-NSF Regional Conference Series in Applied Mathematics, 1972.

\bibitem{Good-Papadopoulos} Good Chris, Papadopoulos Kyriakos, {\sl A topological characterization of ordinals: van Dalen and Wattel revisited}, Topology Appl. 159 (2012), 1565-1572.

    \bibitem{Orderability-Theorem} Kyriakos Papadopoulos, {\sl On the Orderability Problem and the Interval Topology}, Chapter in the Volume ``Topics in Mathematical Analysis and Applications'', in the Optimization and Its Applications Springer Series, T. Rassias and L. Toth Eds, Springer Verlag, 2014.

    \bibitem{Compendium} Gierz, Gerhard and Hofmann, Karl Heinrich and Keimel, Klaus
              and Lawson, Jimmie D. and Mislove, Michael W. and Scott, Dana
              S. A compendium of continuous lattices. {\it Springer-Verlag}, 1980.

\bibitem{Low_path} Robert J. Low, {\sl Spaces of paths and the path topology}, Journal of Mathematical
Physics, 57, 092503 (2016).

\bibitem{Zeeman1} E.C. Zeeman, {\sl The Topology of Minkowski Space},
Topology, Vol. 6, 161-170(1967).

\bibitem{gobel} G\"obel, {\sl Zeeman Topologies on Space-Times of General Relativity Theory}, Comm. Math. Phys. 46, 289-307 (1976).

\bibitem{Zeeman2} E.C. Zeeman, {\sl Causality implies the Lorentz group},
J. Math. Phys. 5 (1964), 490-493.

\bibitem{Ordr-Ambient-Boundary} Ignatios Antoniadis, Spiros Cotsakis and Kyriakos Papadopoulos,
{\sl The Causal Order on the Ambient Boundary}, Mod. Phys. Lett. A, Vol 31, Issue 20, 2016.


\bibitem{Intersection} G.M. Reed, \emph{The intersection topology w.r.t. the real line and the countable ordinals} (Trans. Am. Math. Society, Vol. 297, No 2, 1986, pp 509-520).

\bibitem{Singularities on Amb B} Kyriakos Papadopoulos, \emph{On the possibility of singularities on the ambient boundary} (International
Journal of Geometric Methods in Modern Physics, Vol. 14, No. 10, 2017).



\bibitem{Hawking-Topology} Hawking, S. W. and King, A. R. and McCarthy, P. J. (1976) {\sl A new topology for curved space–time which incorporates the causal, differential, and conformal structures}. Journal of Mathematical Physics, 17 (2). pp. 174-181.

\bibitem{Order-Light-Cone} Kyriakos B. Papadopoulos, Santanu Acharjee and Basil K. Papadopoulos, {\sl The Order On the Light Cone and Its Induced Topology}, International Journal of Geometric Methods in Modern Physics 15, 1850069 (2018).

\bibitem{On-Two-Zeeman-Topologies} Kyriakos B. Papadopoulos and Basil K. Padopoulos, {\sl On Two Topologies that were suggested by Zeeman},
submitted for publication (arxiv: 1706.07488).

\bibitem{OnProperties} Kyriakos Papadopoulos, {\sl On Properties of Nests: Some Answers and Questions}, Questions and Answers in General Topology, Vol. 33, No. 2 (2015).

\bibitem{Nestsandtheirrole} Kyriakos Papadopoulos,  {\sl Nests, and their role in the Orderability Problem}, Mathematical Analysis, Approximation Theory and their Applications, pp 517-533, Th. M. Rassians and V. Gupta Eds, Springer (2016).

\bibitem{Malament} David B. Malament, {\sl The class of continuous timelike curves determines the topology of spacetime},
Journal of Mathematical Physics, 19, 1399 (1977).

\bibitem{Geroch} R. Geroch, {\sl Domain of Dependence}, Journal of Mathematical Physics, 11, 437-449 (1970).

\bibitem{Seth} Seth Major, {\sl  On recovering continuum topology from a causal set}, Journal of Mathematical Physics, 48,
032501, (2007).

\bibitem{Sumati} Onkar Parrikar and Sumati Surya, {\sl Causal Topology in Future and Past Distinguishing Spacetimes}, Classical and Quantum Gravity,
Vol. 28, No 15, 2011.

\bibitem{Rafael} R.D. Sorkin, {\sl Does Locality Fail at Intermediate Length-Scales?}, Approaches to Quantum Gravity: Toward a New Understanding
of Space, Time and Matter, D. Oriti (Ed.), pp. 26-43, Cambridge Univ. Press, Cambridge, UK (2009).

\bibitem{Kalogeropoulos} Nikolaos Kalogeropoulos, {\sl Systolic Aspects of Blackhole Entropy}, Arxiv: 1711.09963

\end{thebibliography}
\end{document}